\documentclass[letterpaper, 10 pt, conference]{ieeeconf}  

\IEEEoverridecommandlockouts                              
\usepackage{amssymb}
\usepackage{paralist}
\usepackage{MnSymbol}
\usepackage{tabularx,booktabs}
\usepackage{graphicx,epstopdf,wrapfig}
\usepackage{comment}
\usepackage{amsmath}
\usepackage{amsfonts}
\usepackage{amssymb}
\usepackage[makeroom]{cancel}
\usepackage{mdwlist,url}
\usepackage{todonotes}

\definecolor{white}{rgb}{1,1,1}
\definecolor{black}{rgb}{0,0,0}
\definecolor{grey}{rgb}{0.7,0.7,0.7}
\definecolor{dgrey}{rgb}{0.5,0.5,0.5}
\definecolor{lightgrey}{rgb}{0.88,0.88,0.88}
\definecolor{lgrey}{rgb}{0.9,0.9,0.9}
\definecolor{llgrey}{rgb}{0.93,0.93,0.93}
\definecolor{lllgrey}{rgb}{0.96,0.96,0.96}
\definecolor{tableHeadGray}{rgb}{0.85,0.85,0.85}
\definecolor{oddRowGrey}{rgb}{0.95,0.95,0.95}
\definecolor{evenRowGrey}{rgb}{0.85,0.85,0.85}
\definecolor{yellow}{rgb}{1.0, 1.0, 0.0}
\definecolor{lightyellow}{rgb}{1.0, 1.0, 0.88}
\definecolor{selectiveyellow}{rgb}{1.0, 0.73, 0.0}
\definecolor{shadered}{rgb}{1,0.85,0.85}
\definecolor{red}{rgb}{1,0,0}
\definecolor{shadegreen}{rgb}{0.95,1,0.95}
\definecolor{green}{rgb}{0,1,0}
\definecolor{darkgreen}{rgb}{0,0.5,0}
\definecolor{shadeblue}{rgb}{0.95,0.95,1}
\definecolor{blue}{rgb}{0,0,1}
\definecolor{darkblue}{rgb}{0,0,0.5}
\definecolor{darkpurple}{rgb}{0.5,0,0.5}
\definecolor{darkdarkpurple}{rgb}{0.3,0,0.3}

\newtheorem{theorem}{Theorem}

\newtheorem{proposition}{Proposition}
\newtheorem{definition}{Definition}

\newtheorem{remark}{Remark}
\newtheorem{example}{Example}

\usepackage{algorithm}
\usepackage{algpseudocode}

\algrenewcommand\algorithmicrequire{\textbf{Input:}}
\algrenewcommand\algorithmicensure{\textbf{Output:}}

\usepackage{listings}

\algnewcommand\algorithmicinput{\textbf{INPUT:}}
\algnewcommand\INPUT{\item[\algorithmicinput]}

\usepackage{tabularx,booktabs}
\usepackage{epstopdf}
\allowdisplaybreaks

\usepackage{tikz}
\usetikzlibrary{arrows,automata}
\usetikzlibrary{calc}
\usetikzlibrary{shapes.geometric,shapes}
\usetikzlibrary{positioning}

\tikzstyle{startstop} = [rectangle, rounded corners, minimum width=3cm, minimum height=1cm,text centered, draw=black]
\tikzstyle{process} = [rectangle, rounded corners,minimum width=3cm, minimum height=.8cm,, text centered, draw=black]
\tikzstyle{decision} = [diamond, minimum width=2cm, minimum height=1cm, text centered, draw=black]
\tikzstyle{arrow} = [thick,->,>=stealth]
\tikzstyle{blank} = [node distance=1cm]
\tikzstyle{block} = [rectangle, draw, fill=blue!20, text centered, rounded corners, minimum height=1em]
\tikzstyle{line} = [draw, -latex']
\tikzstyle{cloud} = [draw, ellipse,fill=red!20, node distance=4cm,minimum height=1em]

\title{\bf  Opacity of Discrete Event Systems with Active Intruder}

\author{Alireza~Partovi~\IEEEmembership{Student~Member,~IEEE,}
        Taeho~Jung,~\IEEEmembership{Member,~IEEE,}
        and~Hai~Lin,~\IEEEmembership{Senior~Member,~IEEE}
        	\thanks{This work was supported in part by the National Science Foundation under Grant IIS-1724070 and  Grant CNS-1830335.}
	\thanks{A. Partovi and H. Lin are with the Department of Electrical Engineering, University of Notre Dame, Notre Dame,	IN, 46556 USA. Emails: {\tt\small apartovi@nd.edu, hlin1@nd.edu}. T. Jung  is with the Department of Computer Science and Engineering, University of Notre Dame, Notre Dame, IN, 46556 USA. E-mail: {\tt\small  tjung@nd.edu}. 
	 }
		}

\tikzstyle{state}=[rectangle,thick,draw=black!75,
   			  fill=black!20,minimum size=2mm]
 \tikzstyle{state_p}=[circle,thick,draw=black!75,
 			  fill=black!20,minimum size=0.8mm,inner sep=0.5mm] 
\tikzstyle{state_obs}=[ellipse,thick,draw=black!75,
 			  fill=black!20,minimum size=0.8mm,inner sep=0.5mm] 
\newcommand*{\LG}{\mathcal{L}(G)}

\begin{document}
\maketitle
\thispagestyle{empty}

\pagestyle{empty}

\begin{abstract} 
Opacity is a security property formalizing the information leakage of a system to an external observer, namely intruder. The conventional opacity that has been studied in the Discrete Event System (DES) literature usually assumes passive intruders, who only observe the behavior of the system. However, in many cybersecurity concerns,  such as web service, active intruders, who are capable of influencing the system's behavior beyond passive observations, need to be considered and defended against. 
We are therefore motivated to extend the opacity notions to handle active intruders. For this, we model the system as a non-deterministic finite-state transducer. It is assumed that the intruder has a full knowledge of the system structure and is capable of interacting with the system by injecting different inputs  and observing its responses. In this setup, we first introduce reactive current-state opacity (RCSO) notion characterizing a property that the system does not leak its secret state regardless of how the intruder manipulates the system behavior. 
We furthermore extend this notion to language-based and initial-state reactive opacity notions, and study the relationship among them. It turns out that all the proposed reactive opacity notions are equivalent to RCSO. We therefore focus on RCSO and study its verification problem.
It is shown that the RCSO can be verified by constructing an observer automaton. 
Illustrative examples are provided throughout the paper to demonstrate the key definition and the effectiveness of the proposed opacity verification approach. 

\end{abstract}

\section{ Introduction }
Cybersecurity is increasingly becoming a great concern as networks of embedded-systems and computers are integrated into almost all aspects of our daily life and society. 
Exchanging confidential information over these networks is crucial in many applications, ranging from smart phones  and home automation to banking services. This  raises a serious  concern on the vulnerability of these systems. 


Many efforts have been made to develop  reliable and secure systems that led to  various notions of security/privacy. One class of security/privacy notations is related to \textit{Information flow} from the system to an external  observer \cite{focardi1994taxonomy}.  \textit{Opacity}  is a type of  information-flow  property that characterizes whether the system's secret information can be inferred by an external observer termed intruder with potentially malicious intentions  \cite{lin2011opacity}.  It is usually
assumed that the intruder knows the system's structure but has only partial observation over its behavior \cite{jacob2016overview}. The system is considered  to be \textit{opaque}  if the intruder is not able to unambiguously determine  the system secrets from its observations.

In recent years, opacity has been extensively studied in the discrete event system (DES) literature, and different  notions of opacity have been proposed, including current-state opacity \cite{saboori2007notions}, language-based opacity \cite{lin2011opacity}, initial-state opacity \cite{saboori2013verification}, $K-$step, and  infinite-step opacity \cite{YIN2017162}.
Interested readers may refer to \cite{jacob2016overview} for a comprehensive review on various notions of opacity.



It is worthy pointing out that the intruder model considered in these methods is a passive observer who is only able to partially observe the system behavior. However, many real-world systems are interacting with malicious and hostile environments, whose capability is beyond a passive observation. A system's malicious environment can act as an \textit{active intruder}, who strategically injects a certain input to the system and observers the system's response to infer its secret. For instance, web browsers and client-side web applications are typical cases of such systems since they interact with remote and possibly untrusted clients that raise a serious concern about the privacy of local users' data \cite{bohannon2009reactive}.

In this paper, we aim at extending the opacity notion in the presence of an active intruder. In particular, who  is  capable of manipulating the system's input and partially observing the system output. 
This setup  naturally models  reactive systems \cite{partovi2019reactive}, 
such as interactive programs \cite{o2006information} and web services \cite{bohannon2009reactive}, 
where input provided by the environment (possibly intruder) and the output of the system is exchanged continuously throughout the  indefinite execution of the system. 

Toward this aim, we introduce \textit{reactive current-state opacity} (RCSO) characterizing the active intruder's ability in manipulating the system's input to certainly determine if the system's current-state is a secret state.
We furthermore extend this notion to \textit{reactive language-based opacity} and \textit{reactive initial-state opacity}. Reactive language-based opacity  requires the secret  behavior of the system to be indistinguishable from a non-secret one. Reactive initial-state opacity notions ensure the active intruder cannot unambiguously determine if the system starts from a secret initial-state. Upon these opacity notions, we present their relationship, the feasibility of each notion, and a procedure to transform one to the other. 
 It turns out that all the proposed reactive opacity notions are equivalent to RCSO. We therefore focus on RCSO, and we study its verification problem.
 
Formal verification of current-state opacity is addressed in \cite{bryans2005opacity} and is further extended to other notions of opacity in \cite{saboori2007notions,saboori2013verification}.  In analogs to verification of opacity with the passive intruder, here we propose to construct an observer automata.  Given the intruder choice of input and the system response (the observable output event), the observer states capture the estimated current-state of the system. Hence, the RCSO verification problem can be reduced to finding the observer states that include a singleton of the secret states.

The contribution of this paper  can be summarized as follows. (i) Consider a new intruder model who has the capability of injecting input into the system; (ii) associated with the new intruder model, we introduce a new class of opacity definitions including the reactive current-state, reactive initial-state, and reactive language-based opacity notions and studies the relationship among them; (iii) provide necessary and sufficient conditions for verification of reactive current-state opacity.

\section{Related Notations} \label{sec:pre}
In this section, we review some preliminary notations that will be used throughout the paper.
For a given finite set (alphabet) of {\it events} $\Sigma$, a finite {\it word} $w=\sigma_1\sigma_2 \ldots \sigma_n$, $n \ge 1$, is a finite sequence of elements in $\Sigma$,  for all $\sigma_i \in \Sigma$, and $1 \le i \le n$. We denote  the length of $w$ by $\vert w \vert$. { Let $w$, and $u$ be finite words, $w \cdot u$ is their \textit{concatenations}.}
The notation $2^\Sigma$ refers to the power set of $\Sigma$, that is, the set of all subsets of $\Sigma$.
A set difference is  $\Sigma-A=\{x \mid x \in \Sigma, x \not \in A\}$. The \textit{free monoid} $\Sigma^*$ generated by $\Sigma$ is the set of all finite sequences $\sigma_1\sigma_2 \ldots \sigma_n$, including the empty sequence denoted by $\epsilon$. A subset of $\Sigma^*$ is called a {\it language} over $\Sigma$.
The {\it prefix-closure} of a language $\mathcal{L} \subseteq \Sigma^*$, denoted as $\overline{\mathcal{L}}$, is the set of all {\it prefixes} of words in $\mathcal{L}$, i.e., $\overline{\mathcal{L}}=\{s\in\Sigma^*|(\exists t\in\Sigma^*)[st\in \mathcal{L}]\}$. $\mathcal{L}$ is said to be {\it prefix-closed} if $\overline{\mathcal{L}}=\mathcal{L}$. 
Let's consider alphabet sets $X$, $Y$, and their set product $\Sigma_{XY}=X \times Y$. A relation $R$ over sets $X$ and $Y$ is a subset of the Cartesian product $X \times Y$. A regular (or rational) relation over the alphabets $X$ and $Y$ is formed from
a finite combination of the following rules:
1: $(x,y) \in (X \cup \{\epsilon \}) \times  (Y \cup \{\epsilon\}$), 2: $\emptyset$ is a regular relation, and 3: If $R_1$, $R_2$ are regular relations, then so are $R_1\cdot R_2$, $R_1 \cap R_2$, and $R_1^*$.
 Projection function to  sets $X$ and $Y$  are respectively denoted as  $\mathrm{P}_{X}=\Sigma^*_{XY} \to X^*$, $\mathrm{P}_{Y}=\Sigma^*_{XY} \to Y^*$, and 
  inductively are defined by $\mathrm{P}_X((\epsilon,\epsilon))=\epsilon$, and 
$\forall w \in \Sigma^*_{XY} $, and $(x,y)\in \Sigma^*_{X Y}$, we have $\mathrm{P}_X(w \cdot (x,y))=\mathrm{P}_X(w) \cdot x$, and   $\mathrm{P}_Y(w \cdot (x,y))=\mathrm{P}_Y(w) \cdot y$. 

A non-deterministic finite state automata (NFA) $A=(Q,\Delta,Q_{0},T_a)$ is a 4-tuple composed of finite state $Q$, a finite set of event $\Delta$, a partial state transition function $T_a: Q \times \Delta \to 2^{Q}$, and the set of  initial states $Q_{0}$. The transition function $T_a$ can be extended to word in a standard recursive manner.
The behavior of NFA
$A$ is captured by $\mathcal{L}(A)=\{s \in  \Delta^* \mid \exists q_0 \in Q_{0} \text{ s.t. } T_a (q_0,s) \neq \emptyset\}$, and for a given initial state $q_0 \in Q_0$ is $\mathcal{L}(A,q_0)=\{s \in \Delta^* \mid T_a (q_0,s) \neq \emptyset\}$. 
$A$ is called deterministic finite automata (DFA) if for any $q \in Q$ and $\delta \in \Delta$ that $T(q,\delta)$ is defined, $|T_a(q,\delta)| = 1$. 

\section{Open Discrete Event System} \label{sec:open_dec}
The finite-state transducers capture transformation of data that is realized by processing inputs and producing outputs using finite memory \cite{mohri2004weighted}. 
We use non-deterministic finite-state transducer (NFT) to characterize the interaction between the system and its environment.
Throughout this paper, we refer to NFT as an \textit{open DES} to emphasize a system model which receives input from an active intruder. 
\begin{definition}[Non-deterministic Finite-State Transducer] \label{def:NFT}
The   nondeterministic finite-state transducer is defined by $G=(Q,X,\Delta, Q_{0},T,\lambda)$, where $Q$ is the finite set of states, $X$ is  finite set of external events, $\Delta=\Delta_{o} \cup \Delta_{uo}$, is the finite set of output events which is partitioned to two disjoint sets of observable output events $\Delta_o$ and unobservable output events $\Delta_{uo}$.  $Q_0$ is the set  of  initial states. The state transition function is $T: Q \times X_\epsilon  \to 2^{Q}$, and $\lambda: Q \times X_\epsilon  \to 2^{\Delta_\epsilon}$ is the output function, where $X_\epsilon= X \cup \{\epsilon \}$ and $\Delta_\epsilon= \Delta \cup \{\epsilon \}$.
\end{definition}
The notation $T(q,x)!$ means that $T(q,x)$ is defined for $x \in X$ and state $q \in Q$.
The extension of $T$ to words is denoted as $T^*:Q \times X^* \to 2^{Q}$ and can be defined recursively for all $q \in Q$ as
$ T^*(q,w)=q$ if $w= \epsilon$, and $T^*(q,w)=\bigcup_{q'\in T(q,x)} T^*(q',v)$ if $w=x\cdot v, x \in X,$ and  $v\in X^*$ \cite{khalili2014learning}.
Here, $T(q,\epsilon)=q$ for each $q \in Q$, indicates that if the input is the empty word, we will remain at the current state.
The extension of output function to words also is denoted as $\lambda^*:Q \times X^* \to 2^{\Delta^*}$, and it can be defined as follows. Given any $w \in X^*$, and $s \in \Delta^*$, we have $s \in \lambda^*(q,w)$  for some $q\in Q$, if and only if, either $w=s=\epsilon$, or  $w=x \cdot w' $, $s= \delta \cdot s'$ for some $x \in X$, and $\delta \in \Delta$, and there exists a state $q' \in Q$ such that $q' \in T(q,x)$, $\delta \in \lambda(q,x)$, and $s' \in \lambda^*(q',w')$. 
The recognized language of $G$ is $\mathcal{L}(G,Q_0)=\{w \in X^* \mid \exists q_0 \in Q_0 \text{ s.t } T(q_0,w)!  \}$. 
Throughout the paper, we use $T$ as a  shorthand for $T^*$,  $\lambda$  for $\lambda^*$, and $\LG$  for $\mathcal{L}(G,Q_0)$.

Given an input word $w \in \LG$, 
 the output word will not be  uniquely determined, due to the non-determinism of the transition and output functions.
For each $q_0 \in Q$ and  $w \in \mathcal{L}(G,q_0)$, a set $O(w,q_0)$ of possible output words is defined inductively as follows:
\begin{itemize}
    \item $O(\epsilon,q_0)=\{ \epsilon\}$,
    \item $\forall w \in \mathcal{L}(G,q_0)$, $\forall x \in X$, such that $w \cdot x \in \mathcal{L}(G,q_0)$: \\
    $O(w \cdot x,q_0)= \{s \cdot \delta \in \Delta^* \mid s \in O(w,q_0) \text{ and } \delta \in  \bigcup_{q \in T(q_0,w)}\lambda(q,x) \}$.
\end{itemize}
We denote $O(w)=\bigcup_{q_0 \in Q_0}O(w,q_0)$. The set of all possible output words in $G$ is denoted by $O(\LG)$, that is, $O(\LG)=\bigcup_{q_0 \in Q_0, w \in \mathcal{L}(G,q_0)} O(w,q_0) \subseteq \Delta^*$. We call $O(\LG)$ the output language of $G$.
\begin{figure}[t]  
\centering \vspace{5pt}
\begin{tikzpicture}[shorten >=1pt,node distance=1.5cm,on grid,auto, bend angle=20, thick,scale=1, every node/.style={transform shape}] 
	\node[state_p,initial left,initial text=] (q_0)   {$\scriptstyle 0$};
    \node[state_p] (q_1) [right=of q_0,xshift=0cm,yshift=1.5cm] {$\scriptstyle 1$};
    \node[state_p] (q_2) [right=of q_0,xshift=1.7cm,yshift=0cm] {$\scriptstyle 2$};
    \node[state_p] (q_3) [right=of q_0,yshift=-1.5cm] {$\scriptstyle 3$};f
	\path[->]
	(q_0) edge [] node [sloped]     { $\scriptstyle x_1\slash \{\delta_1 \delta_2\}
	$} (q_1)
	(q_0) edge [] node [above, align=center, pos=0.4]     { $\scriptstyle x_2\slash \{\delta_2\}$} (q_2)
	(q_0) edge [] node [sloped,below, align=center, pos=0.5]  { $\scriptstyle x_1\slash \{\delta_1 \delta_2\}$} (q_3)
	(q_1) edge [] node [sloped, below,align=center, pos=0.4]     { $\scriptstyle x_1\slash \{\delta_2\}$} (q_2)
	(q_1) edge [loop above] node [ right ]     { $\scriptstyle x_2\slash \{b,a\}$} (q_1)
	(q_2) edge [loop right] node [ pos=0.7]     { $\scriptstyle x_1\slash \{\delta_2, a \}, x_2\slash \{b\} $} (q_2)
	(q_2) edge [bend left=40] node [below,sloped,align=center, pos=0.5 ]     { $\scriptstyle x_2\slash \{\delta_1\}$} (q_3)
	(q_2) edge [bend right=40] node [above,sloped,align=center, pos=0.5 ]     { $\scriptstyle x_2\slash \{\delta_1\}$} (q_1)
	
	(q_3) edge [] node [sloped,above,align=center, pos=0.4] { $\scriptstyle x_1\slash \{\delta_2\}$} (q_2)

	(q_3) edge [loop below] node [above left] { 
	     $\scriptstyle x_1\slash \{\delta_2\}$,
	     $\scriptstyle x_2\slash \{a\}$
} (q_3)
    ;
\end{tikzpicture}\vspace{-10pt}
\caption{An example of open DES $G$. Note that $\epsilon \in \lambda(q,x)$ for all $q \in Q$, and $x \in X_\epsilon$.
We removed the $\epsilon$ input transitions for clarity of the figures.
}\vspace{-10pt}
\label{fig:example_NFT}  
\end{figure}
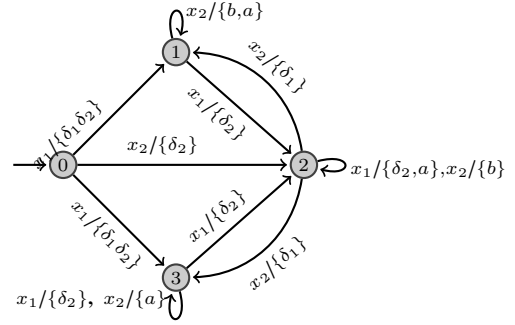 
\begin{example} \label{exp:open_des}
Consider the open DES shown in Figure \ref{fig:example_NFT}, where  $\Delta=\{\delta_1,\delta_2,a,b\}$,  $X=\{x_1,x_2\}$, and the initial state is $Q_0=\{0\}$.
An edge in the model is in the form of  $x \slash Y$, where $x \in X_\epsilon$, represents the  input event, and   and $Y \subseteq \Delta_\epsilon$ denotes the set of possible output events. Multiple labels over an edge indicates  multiple enabled transitions. 
For instance, for $x_1 x_1 \in \LG$, we have $O(x_1x_1)=\{\delta_1 \delta_2, \delta_2 \delta_2\}$, that is, two output words, $\delta_1 \delta_2$, and $\delta_2 \delta_2$ are possible. 
\hfill $\Box$
\end{example}

If there are  marked states, we define open DES as  $G=(Q,X,\Delta, Q_{0},T,\lambda,F)$, where $F \subseteq Q$ are the marked states.
The input-output  language of $G$, denoted as $\mathcal{L}_{io}(G)$, is defined by  $\mathcal{L}_{io}(G)=\{(w,s) \in (X \times \Delta)^* \mid \exists q_0 \in Q_0, \text{ s.t. } T(q_0,w)!, \text{ and } s \in \lambda(q_0,w)\}$, and its   input-output marked  language is given by $\mathcal{L}_{io,m}(G)=\{(w,s) \in (X \times \Delta)^* \mid \exists q_0 \in Q_0, \text{ s.t. } T(q_0,w) \cap F \neq \emptyset, \text{ and } s \in \lambda(q_0,w)\}$.
The input-output languages of $G$ is a {regular relation} over the set $(X \cup \{\epsilon\}) \times (\Delta \cup \{ \epsilon\})$ that  can be conveniently recognized by an non-deterministic finite-state transducer \cite{bouajjani2000regular}.

The \text{accessible} part of an NFT $G=(Q,X,\Delta, Q_{0},T,\lambda,F)$ is denoted by $Ac(G)$ and is obtained by removing the states that cannot be reached from any initial state $q_0 \in Q_0$ in finite number of steps. The coaccessible part of $G$, denoted by $CoAc(G)$ is an NFT obtained by deleting the states that cannot reach to the marked states $F$. The trim operation, denoted by ${Trim}$, transforms $G$ to another NFT as a part of $G$ that is both accessible and coaccessible, formally ${Trim}(G)=Ac(CoAc(G))=CoAc(Ac(G))$ \cite{cassandras2009introduction}. Similarly, for an NFA $A$, we can define $Trim(A)$, $Ac(A)$, and $CoAc(A)$.

\section{Opacity Of Discrete-Event Systems} \label{sec:opacity_def}
Opacity is characterized by the system's secret and the intruder's observation mapping over the system's executions. The system is \textit{opaque}, if for any execution run that contains secret, there exists another non-secret run which is observably equivalent.
In the  formalism of  opacity, the intruder is  considered as  an \textit{observer} who has  full knowledge of  the system structure but has a partial observability  over it. Typically, the intruder's partial observability  is modeled by a \textit{natural projection}  function. 
The natural projection is $P:\Delta^* \to \Delta^*_o$, and for any $s\in \Delta^*$, and $\delta \in \Delta$, it is defined recursively by $P(\epsilon)=\epsilon$, and $P(s \cdot \delta ) =P(s) \cdot \delta$ if $\delta  \in  \Delta_o$ and otherwise $P(s \cdot \delta )=P(s)$.

The system secret information or behavior can be represented in different ways, such as  secret states and languages. 
In the conventional opacity of DESs with passive intruder, various opacity notions for different representation of secret have been introduced including but not limited to current-state, language-based, and initial-state opacity \cite{jacob2016overview}. 

\subsection{Current-State Opacity}
Here, we first discuss the   current-state opacity (CSO) definition when the intruder is just a passive observer; and later, we will show how an active intruder can force a current-state opaque system to expose its secret states. 

\begin{definition}[Current-State Opacity]
Given a non-deterministic finite-state automata  $A=(Q,\Delta,Q_{0},T_a)$, and a passive intruder with projection function $P$, a set of secret state $Q_s \subset Q$, the system $A$ is \textit{current-state opaque}  if  $\forall q_0 \in Q_0$ and  $\forall s \in \mathcal{L}(A,q_0)$ such that $T_a(q_0,s) \subseteq Q_s$, there exists $q'_0 \in Q_0$ and $\exists s' \in \mathcal{L}(A,q'_0)$, such that $T_a(q'_0,s') \subseteq \{Q - Q_s\}$ and $P(s)=P(s')$.
\end{definition}
Intuitively, when the intruder  can only observe the system outputs with projection $P$,  $A$ is current-state opaque  if for  every word $s \in \mathcal{L}(A)$ leading to a secret state in $Q_s$, there exists at least  another word $s' \in \mathcal{L}(A)$ that leads to non-secret states $\{Q - Q_s\}$ whose projection is the same. Thus, the intruder can never  determine that the system's current state is in $Q_s$. 
One can check whether the system $A$ with a passive intruder is current-state opaque by constructing a \textit{current-state estimator} (observer) and by verifying that no (nonempty) current-state estimate lies entirely within the set of secret states $Q_s$ \cite{hadjicostis2014opacity}.



\begin{example} \label{exp:passive_active_intruder}
Consider the open DES $G$ depicted in Figure \ref{fig:example_NFT} with  $\Delta_{o}=\{\delta_1,\delta_2,a\}$,  $\Delta_{uo}=\{b\}$,  and $Q_{s}=\{3\}$. We first assume the intruder is passive and can only observe the  observable outputs through projection function $P$. In order to evaluate CSO on $G$,  we can associate  a NFA $A$ with  the open DES $G$. Let's consider the NFA $A_G=(Q,\Delta,Q_{0},T_a')$, where the transition function $T_a'$, for any $q,q' \in Q$, and $\delta \in \Delta$, is defined as $q' \in T_a'(q,\delta)$, if there exists $x \in X$ such that $q' \in T(q,x)$ and $\delta \in \lambda(q,x)$; otherwise $T_a'(q,\delta)$ is not defined. We can construct an observer automata to check if $A_{G}$ is current-state opaque with respect to $P$, and $Q_s$. The observer is shown in Figure \ref{fig:example_output_obs}. The observer shows the secret state $\{3\}$ never lies entirely on single state of the observer, and hence, $A$ is current-state opaque with respect to $Q_s$ and $P$. However, if the intruder is capable of providing a certain input word to the system and observe the system's output through $P$, she can infer when the system is in the secret state. Specifically, consider the input word $w=x_1 x^{*}_2 x_1$ that drives the system to land on one of the states $\{2,3\}$, and here, if the active intruder chooses $x_2$, i.e., $w \cdot x_2$ and observes $a$, she can infer the current-state of the system is certainly at  the secret state $\{3\}$. However, if $a$ is an unobservable event, the active intruder with the same input word $x_1 x^{*}_2 x_1 x_2$, cannot  determine whether the system is at $\{3\}$ or $\{2\}$.  \hfill $\Box$
\end{example}
\begin{figure}[t]  
\centering \vspace{10pt}
\begin{tikzpicture}[shorten >=1pt,node distance=2.0cm,on grid,auto, bend angle=20, thick,scale=1, every node/.style={transform shape}] 
	\node[state_obs,initial left,initial text=] (q_0)   {$\scriptstyle \{0\}$};
    \node[state_obs] (q_123) [above right =of q_0,yshift=0cm] {$\scriptstyle \{1,2,3\}$};
    \node[state_obs] (q_13) [right  =of q_0,xshift=1.0cm,yshift=0cm] {$\scriptstyle \{1,3\}$};
    \node[state_obs] (q_23) [right=of q_123,xshift=1cm] {$\scriptstyle \{2,3\}$};
	\path[->]
	(q_0) edge [] node [below,sloped,align=center]{ $\scriptstyle  \delta_2$} (q_123)
	(q_0) edge [] node [sloped]{ $\scriptstyle  \delta_1$} (q_13)
	(q_123) edge [loop left] node []     { $\scriptstyle  a$} (q_123)
	(q_123) edge [] node [sloped,align=center]     { $\scriptstyle  \delta_2$} (q_23)
	(q_123) edge [] node [sloped, above]     { $\scriptstyle  \delta_1$} (q_13)
	(q_23) edge [loop right] node []     { $\scriptstyle  \delta_2,a$} (q_23)
	(q_23) edge [bend left=40] node [sloped, below]     { $\scriptstyle  \delta_1$} (q_13)
	(q_13) edge [loop above] node []     { $\scriptstyle  a$} (q_13)
	(q_13) edge [] node [below,sloped,pos=0.5]     { $\scriptstyle  \delta_2$} (q_23)
    ;
\end{tikzpicture}\vspace{-5pt}
\caption{Current-state estimator of the  passive intruder for the open DES in Figure \ref{fig:example_NFT}.}
\label{fig:example_output_obs}\vspace{-10pt}
\end{figure}
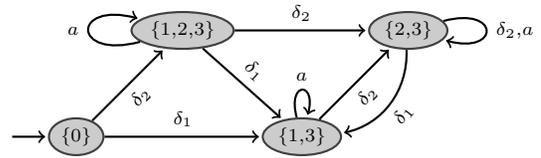 
As Example \ref{exp:passive_active_intruder} illustrates, an active intruder can force the open DES $G$ to expose his secret-state. We, therefore, need a new current-state opacity notion that captures this active intruder ability. In particular, we consider an active intruder who has full knowledge of the open DES model; and is capable of injecting input to the system and  (partially) observing the system output.

To evaluate an open DES current-state opacity, we can construct a current-state estimator that tracks the active intruder estimated states.
Given an input word accepted by the system $w \in \LG$, and an observed word $\alpha \in P(O(\LG))$, the  current-state estimator is defined by:
\begin{align*}
\Tilde{Q}^G(w,\alpha)= \{ q \in Q \mid& \exists q_0 \in Q, q\in T(q_0,w), \text{ and } \\
& \exists s \in O(w,q_0), \text{ s.t. }P(s)=\alpha  \}.
\end{align*}

The current-state estimator $\Tilde{Q}^G(w,\alpha)$ essentially characterizes a set of states which the open DES lands on as a result of the input word $w$, and meanwhile it produces the observable sequences $\alpha$.
We also define the current-state estimator for a given initial state $q_0 \in Q_0$, as $\Tilde{Q}^G_{q_0}(w,\alpha)=\{ q \in Q \mid q\in T(q_0,w), \text{ and } \exists s \in O(w,q_0), \text{ s.t. }P(s)=\alpha  \}$. We use $\Tilde{Q}(w,\alpha)$ instead of $\Tilde{Q}^G(w,\alpha)$, and $\Tilde{Q}_{q_0}(w,\alpha)$ for $\Tilde{Q}^G_{q_0}(w,\alpha)$, when it is clear from the context. 
Upon this current-state estimator, we define the reactive current-state opacity in the following. 



\begin{definition}[Reactive Current-State Opacity]
Given an open DES $G=(Q, \allowbreak X,\Delta, Q_{0},T,\lambda)$, projection function $P$, and the set of secret states $Q_s \subset Q$, the system is \textit{reactive current-state opaque} (RCS-opaque)  if for any $w \in \LG$ there exists $q_0 \in Q_0$ such that: 
\begin{itemize}
    \item  $T(q_0,w) \cap \{Q - Q_s \} \neq \emptyset$,
    \item $\forall t \in P(O(w,q_0))$, we have $\Tilde{Q}_{q_0}(w,t)  \cap \{Q - Q_s \} \neq \emptyset$.
\end{itemize}
\end{definition}

Intuitively, the open DES $G$ is RCS-opaque, if with any input word $w$ that is recognized by $G$, i.e., $w\in \LG$, i) there exists an initial state $q_0 \in Q_0$ such that  the system with $w$ does not land entirely at the secret states, i.e., $T(q_0,w) \cap \{Q- Q_s\} \neq \emptyset$; and  ii) for any possible observable output word associated with the input, $t \in P(O(w,q_0))$, we have $\Tilde{Q}_{q_0}(w,t) \cap \{Q - Q_s \} \neq \emptyset$, that is, the intruder cannot use the observed output events to resolve the non-determinism of the transition function $T(q_0,w)$ to infer the current secret state of the system.

\begin{remark}
In the definition of RCSO, the input word $w$, is not required to be restricted to the recognized words by the open DES $G$, $w \in \LG$, and it can be any $ w\in X^*$. However, clearly $G$ does not accept any $w\in \{X^* - \LG\}$, and hence, it does not reveal any secret.
\end{remark}

\begin{example} \label{exp:RCSO}
Consider the system $G$ in Figure \ref{fig:example_NFT}, with secret state set $Q_s=\{2\}$. In this case, $G$ is not RCSO since the intruder with input word $w=x_2$, and regardless of the observed output events,  can ensure the system current-state is $\{2\}$. However, if $Q_s=\{3\}$, the system with any $w \in \LG$, does not proceed solely to $Q_s$, and therefore, the intruder potentially  can use the observed output events to infer the secret state from the system's possible current-states.  For instance, with $x_1 x_1x_2$,  the  possible current-states of the system are $\{2,3 \}$, and if the observed output word is $t \cdot a$, where $t$ is any  $t\in O(x_1 x_1)$,  the intruder is able to certainly infer  the  current-state of $G$ is the secret state $\{3\}$, that indicates $G$ is not  RCS-opaque. \hfill $\Box$
\end{example}

\begin{remark}
The proposed RCSO notion with an active intruder is a generalization of CSO notion with the passive intruder. 
As it is illustrated in Example \ref{exp:passive_active_intruder}, if we consider open DES with a passive intruder who has a partial observation on the system's output,  the proposed RCSO can capture  the CSO notion. 
\end{remark}

\subsection{Other Opacity Notions}
Other notions of opacity  can be extended to the open DESs  with an active intruder. In this paper, we introduce  reactive language-based and reactive initial-state opacity notions. The reactive language-based opacity (RLBO)  characterizes  a secret run of the system that should be protected against an active intruder. 

\begin{definition}[Reactive Language-Based Opacity]
Given an open DES $G=(Q,X,\Delta, Q_{0},T,\lambda)$, projection function $P$, and secret output language $O_s \subset O(\LG)$, and non-secret output language $O_{ns} \subseteq O(\LG)$,  $G$ is reactive language-based opaque, if  for all $q_0 \in Q_0$, and any $ w \in \mathcal{L}(G,q_0)$  that $O(w,q_0) \cap O_{s} \neq \emptyset$, there exists $q'_0 \in Q_0$ such that:
\begin{itemize}
    \item $O(w,q'_0) \cap O_{ns} \neq \emptyset$,
    \item $\forall  t \in (O(w,q_0) \cap O_{s}), \exists t' \in (O(w,q'_0) \cap O_{ns})$ such that $P(t)=P(t')$.
\end{itemize}
\end{definition}

Intuitively, $G$ is reactive language-based opaque with respect to the secret output language $O_s$, non-secret output language $O_{ns}$, and the projection function $P$, if for any input word $w \in \mathcal{L}(G,q_0)$ that generates secret output word, $O(w,q_0) \cap O_{s} \neq \emptyset$, there exists an initial state $q'_0 \in Q_0$, such that the same input word from the intruder can be associated with a non-secret output word, $O(w,q'_0) \cap O_{ns} \neq \emptyset$, and additionally, 
for any secret output word $t \in O(w,q_0) \cap O_{s}$
 there exists a non-secret output word $t' \in (O(w,q'_0) \cap O_{ns})$, such that they have the same observation $P(t)=P(t')$. 

Initial-state opacity is another notion of opacity defined over the system secret initial states. 
For open DESs, reactive initial-state opacity (RISO) can be defined as follows. 

\begin{definition}(Reactive Initial State Opacity)
 Given an open DES $G=(Q, X,\Delta,\allowbreak Q_{0},T, \lambda)$, projection function $P$, and secret initial state set $Q^0_s \subset Q_{0}$, and non-secret initial state set $Q^0_{ns} \subseteq Q_{0}$,  $G$ is reactive initial-state opaque, if $\forall q_0 \in Q^0_s$ and any input words $w \in \LG$ with any $t \in O(w,q_0)$, there exists a non-secret initial-state $q'_0 \in Q^0_{ns}$ and  $ t' \in O(w,q'_0)$ such that $P(t)=P(t')$.
 
\end{definition}
An open DES $G$ is reactive initial-state opaque with respect to the secret initial-state set $Q^0_s$,  non-secret initial-state set $Q^0_{ns}$, and the projection function $P$, if for any secret initial-state $q_0 \in Q^{0}_s$, and any input word $w \in \mathcal{L}(G)$, that generates an output word $t$, i.e., $t\in O(w,q_0)$, there exists a non-secret initial state $q'_0 \in Q^0_{ns}$, and an output word $t' \in O(w,q'_0)$, associated with   $w$ and  $q'_0$,  such that, $t$ and $t'$ have the same observation, i.e., $P(t)=P(t')$.

Similar to the opacity notions with a passive intruder \cite{wu2013comparative}, there is a relationship between the proposed reactive opacity notions. We call a problem of checking if a given open DES satisfies the RCSO conditions, a RCSO problem. Similarly, in the sequel, we use the terms RLBO and RISO problems. We mainly follow the idea proposed in \cite{wu2013comparative} to transform the reactive  opacity problems to each other.

\begin{proposition} \label{prop:RLBO2RCSO}
A RLBO problem can be converted to an equivalent RCSO problem.
\end{proposition}

\begin{proof}
Construct  an NFT $G_s=(S_s,X,\Delta,T_s,S_{s0},\lambda_s,F_s)$ such that $\mathcal{L}_{io,m}(G_s)=\{(w,s) \in (X \times \Delta)^* \mid w \in \mathcal{L}(G)\text{ and } s\in O_s \}$, and  an NFT $G_{ns}=(S_{ns},X,\Delta,T_{ns},S_{ns0},\lambda_{ns},F_{ns})$ that accepts  $\mathcal{L}_{io,m}(G_{ns})=\{(w,s) \in (X \times \Delta)^* \mid w \in \mathcal{L}(G) \text{ and } s\in O_{ns} \}$. Then consider $G_s$ and $G_{ns}$ as single NFT by constructing $G_c=(S_s \cup S_{ns},X,\Delta,T_s \cup T_{ns},S_{s0} \cup S_{ns0},\lambda_s \cup \lambda_{ns} ,F_s \cup F_{ns})$, and define the secret and non-secret state sets respectively as $Q_s=F_s$ and  $Q_{ns}=F_{ns}$.
Therefore, for any $q_0 \in Q_0$, $w \in \mathcal{L}(G,q_0)$ and $t\in O(w,q_0) \subseteq O_s$, there  exist $s_0 \in (S_{s0} \cup S_{ns0})$ and $\rho \in \mathcal{L}_{io,m}(G_c,s_0)$ with $P_X(\rho)=w$ and $P_{\Delta_o}(\rho)=t$, such that $\Tilde{Q}^{G_c}_{s_0}(w,t) \subseteq Q_s$; and  if $\exists q'_0 \in Q_0 $ and $t \in O(w,q'_0) \subseteq O_{ns}$, indicating $G$ is reactive language-based opaque, we  have $s'_0 \in (S_{s0} \cup S_{ns0})$ and $\rho'  \in \mathcal{L}_{io,m}(G_c,s'_0)$  with $P_X(\rho')=w$ and $P_{\Delta_o}(\rho')=t'$, such that $\Tilde{Q}^{G_c}_{s'_0}(w,t') \subseteq Q_{ns}$, which implies $G_c$ is RCS-opaque.
\end{proof}

The other direction of this transformation is also possible. A RCSO problem can be converted to an equivalent RLBO problem.

\begin{proposition} \label{prop:RCSO2RLBO}
A RCSO problem can be converted to an equivalent RLBO problem.
\end{proposition}
\begin{proof}
Given an RCSO problem with $G=(Q,X,\Delta, Q_{0},T,\lambda)$,  secret states $Q_s \subset Q$, and non-secret  states set  $Q_{ns} \subseteq Q$. Construct an NFT  with $Q_s$ as the marked states, defined as $G_s={Trim}(Q,X,\Delta,T,Q_{0},\lambda,Q_s)$, and another NFT with $Q_{ns}$ as the marked states, given by  $G_{ns}={Trim}(Q,X,\Delta,T,Q_{0},\lambda,Q_{ns})$. Then define  the secret and non-secret output language respectively by $O_{s}= P_{\Delta}(\mathcal{L}_{io,m}(G_{s}))$ and $O_{ns}=P_{\Delta} (\mathcal{L}_{io,m}(G_{ns}))$. 
\end{proof}

The RISO is  related to the RLBO. Proposition \ref{prop:RISO2RLBO} and \ref{prop:RLBO2RISO} establish this relationship. 

\begin{proposition} \label{prop:RISO2RLBO}
RISO problem can be converted  to an equivalent RLBO problem.
\end{proposition}

\begin{proof}
Given open RISO problem with  $G=(Q,X,\Delta,T,Q_{0},\lambda)$,  secret initial-state set $Q^0_s \subset Q_0$, and non-secret initial  state set $Q^0_{ns} \subseteq Q_{0}$, construct an NFT by trimming $G$ to only the secret initial-state set $Q^0_s$, given as $G_s = Trim(Q,X,\Delta,T,Q^0_{s},\lambda)$, and similarly  construct another NFT with $Q^0_{ns}$ as initial-state set, $G_{ns} = Trim(Q,\allowbreak X,\Delta,T,Q^0_{ns}, \lambda)$. Then combine $G_s$ and $G_{ns}$ as $G_{l}=Trim(Q,X,\Delta,T,Q^0_{s}\cup Q^0_{ns},\lambda)$, and  define the secret and non-secret output languages respectively by $O_s=O(\mathcal{L}(G,Q^0_{s}))$, and $O_{ns}=O(\mathcal{L}(G,Q^0_{ns}))$.
\end{proof}

The other direction of this transformation does not always hold. A RLBO problem can  be transformed to an equivalent RISO only if $O_s$ and $O_{ns}$ are prefix-closed.
\begin{proposition} \label{prop:RLBO2RISO}
Given a RLBO problem with prefix-closed $O_s$ and $O_{ns}$, there exists an equivalent RISO problem.
\end{proposition}
\begin{proof}
Given an RLBO problem with the open DES $G=(Q,X,\Delta,T,Q_{0},\lambda)$, and  prefix-closed secret output language $O_s \subset O(\mathcal{L}(G))$, and prefix-closed non-secret  output language $O_{ns} \subseteq O(\mathcal{L}(G))$. Construct  an NFT $G_s=(S_s,X,\Delta,T_s,S_{s0},\lambda_s)$ such that $\mathcal{L}_{io}(G_s)=\{\rho \in  (X \times \Delta)^* \mid P_X(\rho) \in \mathcal{L}(G)\text{ and } P_\Delta(\rho)\in O_s \}$, and  an NFT $G_{ns}=(S_{ns},X,\Delta,T_{ns},S_{ns0},\lambda_{ns})$ that accepts  $\mathcal{L}_{io}(G_{ns})=\{\rho \in  (X \times \Delta)^* \mid P_X(\rho)\in \mathcal{L}(G) \text{ and } P_\Delta(\rho)\in O_{ns} \}$. Then consider $G_s$ and $G_{ns}$ as single NFT by constructing $G_c=(S_s \cup S_{ns},X,\Delta,T_s \cup T_{ns},S_{s0} \cup S_{ns0},\lambda_s \cup \lambda_{ns})$, and define the secret and non-secret initial-state sets respectively as $Q^0_s=S_{s0}$ and  $Q^0_{ns}=S_{ns0}$. 
\end{proof}

\begin{remark}
It is shown that the proposed RCSO and RLBO are equivalent properties for $G$. The RISO  can be transformed to a RLBO property, however, the reverse of this transformation (RLBO to RISO), only holds for prefix-closed secret and non-secret languages. Therefore, if the prefix-closed conditions hold, RISO is also an equivalent property to RCSO.
Figure \ref{fig:reactive_opacity_relation} illustrates this relation.
\end{remark}

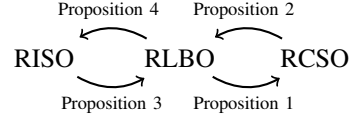
\begin{figure}
\centering \vspace{10pt}
\begin{tikzpicture}[shorten >=1pt,node distance=1.8cm,on grid,auto, bend angle=20, thick,scale=1, every node/.style={transform shape}] 
	\node (RLBO) [ xshift=0cm,yshift=0cm] {RLBO };
	\node (RCSO) [ right of=RLBO, xshift=0cm,yshift=0cm] {RCSO };
	\node (RISO) [ left of=RLBO, xshift=0cm,yshift=0cm] {RISO };
  	\path[->]
    (RLBO) edge [bend right=30]  node [below,sloped,align=center]{\scriptsize Proposition \ref{prop:RLBO2RCSO} }   (RCSO)
     (RCSO) edge [bend right=30]  node [above,sloped,align=center]{\scriptsize Proposition \ref{prop:RCSO2RLBO} } (RLBO)
    (RLBO) edge [bend right=30] node [above,sloped,align=center]{    	\scriptsize  Proposition \ref{prop:RLBO2RISO}
    } (RISO)
    (RISO) edge  [bend right=30]  node [below,sloped,align=center]{\scriptsize Proposition \ref{prop:RISO2RLBO}} (RLBO)
    ;
\end{tikzpicture}\vspace{-10pt}
\caption{The equivalence relation in the reactive opacity notions.}
\label{fig:reactive_opacity_relation} \vspace{-10pt}
\end{figure}

\section{RCSO Verification} \label{sec:RCSO_verf}
In this section, we present the verification of RCSO notion for open DESs. Similar to current-state opacity with a passive intruder \cite{saboori2007notions}, we can construct an observer automata to verify if an open DES is RCS-opaque.  
In conventional opacity with a passive intruder, the observer is constructed to track the system states based on the observable events \cite{hadjicostis2014opacity}.
In the reactive opacity formalism, however, the intruder knows the injected input word, and hence the system (non-deterministic) transitions. As it is illustrated in Example \ref{exp:RCSO}, the active intruder can utilize the system observable responses to resolve the ambiguity of his estimation caused by the system's non-deterministic transition.  The observer for RCSO verification ,therefore, should include both possible input and observable output behavior of the system to track the estimated states.  Furthermore, an open DES may only have a single and perhaps unique unobservable output event for a given input that can reveal a secret state. Therefore,  in contrary to the conventional opacity with passive intruder,  an active intruder can even use an unobservable response to infer the open DES states. This ability should be encoded in the active intruder observer.

\begin{definition}[Observer for RCSO] \label{def:RCSO_observer}
Given an open DES $G=(Q,X,\Delta,Q_0,T,\lambda)$, a projection function $P$ with respect to the observable output events $\Delta_o$, the observer automata is a deterministic finite-state automata $G_o=Ac(\hat{Q},X, \Delta_o,\hat{Q}_0,T_o)$  with state set $\hat{Q}=2^Q$, the initial state set is $\hat{Q}_0=Q_0 \cup \{ q \in Q \mid \exists q_0 \in Q_0, \text{ s.t }  q\in T(q_0,\epsilon) \}$. 
Let's denote $\Delta_{o,\epsilon}=\Delta_o \cup \{\epsilon\}$, the transition function is $T_o: \hat{Q} \times X_\epsilon \times \Delta_{o,\epsilon} \to {\hat{Q}}$, that for any $\hat{q} \in \hat{Q}$, $x \in X_\epsilon$, and an observable event $\delta \in \Delta_o$ is given by $T_o(\hat{q},(x,\delta))=\{\hat{q}' \in \hat{Q}  \mid \exists q \in \hat{q} \text{ s.t } \hat{q}'\subseteq T(q,x) \text{ and } \delta \in \lambda(q,x)\}$, and for an unobservable event, it is defined by
$T_o(\hat{q},(x,\epsilon))=\{\hat{q}' \in \hat{Q}\mid \exists q \in \hat{q} \text{ s.t } \hat{q}' \subseteq T(q,x) \text{ and }  \exists \delta_{uo} \in (\Delta_{uo} \cup \{\epsilon\} ) \text{ s.t. } \delta_{uo} \in  \lambda(q,x)\}$.  


\end{definition}
The initial estimated states $\hat{Q}_0$ is constructed based on the combination of the possible initial states, $Q_0$, and any initial transitions with no input to the open DES, i.e., $T(q_0,\epsilon)$. Note that, based on the definition of open DES in Definition \ref{def:NFT}, for any $q_0 \in Q_0$, we have  $\epsilon \in  \lambda(q_0,\epsilon)$, and therefore, $\hat{Q}_0$ is solely defined based on $Q_0$ and $T(q_0,\epsilon)$. In the constructed observer, $T_o(\hat{q},(x,\epsilon))$ captures the active intruder ability to infer the system transition when he injects input $x$ and receives no observable output.

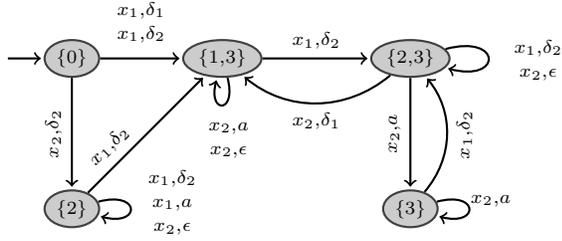
\begin{figure}[t]  
\centering
\begin{tikzpicture}[shorten >=1pt,node distance=2.0cm,on grid,auto, bend angle=20, thick,scale=1, every node/.style={transform shape}] 
	\node[state_obs,initial left,initial text=] (q_0)   {$\scriptstyle \{0\}$};
    \node[state_obs] (q_13) [right  =of q_0,xshift=0.0cm,yshift=0cm] {$\scriptstyle \{1,3\}$};
    \node[state_obs] (q_23) [right=of q_13,xshift=0.5cm] {$\scriptstyle \{2,3\}$};
    \node[state_obs] (q_2) [below=of q_0,xshift=.0cm] {$\scriptstyle \{2\}$};
    \node[state_obs] (q_3) [below=of q_23,xshift=0.0cm] {$\scriptstyle \{3\}$};
	\path[->]
	(q_0) edge [] node [sloped,above, rotate=180]{$\scriptstyle  x_2,\delta_2$ } (q_2)
	(q_0) edge [] node [sloped,above]{ \begin{tabular}{c} $\scriptstyle  x_1,\delta_1$\\[-1mm]$\scriptstyle  x_1,\delta_2$\end{tabular}} (q_13)
	(q_13) edge [] node [sloped,align=center]     { $\scriptstyle  x_1,\delta_2$} (q_23)
	 (q_13) edge [loop below] node [pos=.25]     { \begin{tabular}{c} $\scriptstyle  x_2,a$\\[-1mm]$\scriptstyle  x_2,\epsilon$\end{tabular}} (q_13)
	(q_23) edge [loop right] node []     {  \begin{tabular}{c} $\scriptstyle  x_1,\delta_2$\\[-1mm]$\scriptstyle  x_2,\epsilon$\end{tabular}} (q_23)
 	(q_23) edge [bend left=40] node [sloped, below]     { $\scriptstyle  x_2,\delta_1$} (q_13)
 	(q_23) edge [] node [sloped, above,rotate=180]     { $\scriptstyle  x_2,a$} (q_3)
 	(q_3) edge [loop right] node [pos=.25]     { $\scriptstyle  x_2,a$} (q_3)
 	(q_3) edge [bend right=40] node [sloped, below]     { $\scriptstyle  x_1,\delta_2$} (q_23)
 	(q_2) edge [sloped] node [pos=.5]     { $\scriptstyle  x_1,\delta_2$} (q_13)
 	(q_2) edge [loop right] node [pos=.25]     {\begin{tabular}{c} $\scriptstyle  x_1,\delta_2$\\[-1mm]$\scriptstyle  x_1,a$\\[-1mm]$\scriptstyle  x_2,\epsilon$ \end{tabular}} (q_2)
    ;
\end{tikzpicture}\vspace{-10pt}
 \caption{Observer automata for the open DES in Example \ref{exp:open_des}. For clarity of the figure we remove all the  transitions for the empty input, $x=\epsilon$.} 
\label{fig:example_observer} 
\end{figure} 

Given the constructed observer $G_o$, one can verify if $G$ is RCS-opaque by checking if there exists any state $\hat{q} \in \hat{Q}$ which is reachable from $\hat{Q}_0$ and only contains the system secret states $Q_s$, i.e., $ \hat{q} \subseteq Q_s$. 
The RCSO verification based on the proposed observer construction is formally given in the following theorem. 

\begin{theorem}
Given an open DES $G=(Q,X,\Delta, Q_{0},T,\lambda)$, the projection function $P$, the secret state set $Q_s \subset Q$,  the associated observer  $G_o=Ac(\hat{Q},X,\Delta_o,\hat{Q}_0,T_o)$ can be constructed by following Definition \ref{def:RCSO_observer}. Then  $ G$ is   RCS-opaque if and only if for all $\hat{q} \in \hat{Q}$ either $\hat{q}=\emptyset$ or $\hat{q} \not \subseteq Q_s $ holds.
\end{theorem}

\begin{proof}
Necessary: here we show if $G$ is RCS-opaque, then there is no state $\hat{q} \in \hat{Q}$ in the constructed observer (following Definition \ref{def:RCSO_observer}) that $ \hat{q}\not \subseteq Q_s $. Let's denote  $Q_o \subseteq \hat{Q}$ as the reachable states in $G_o$.
To prove this part, we only need to show that for any input word  and the observed output word, the states in the observer $G_{o}$ are the estimated current-state of the system. Consider any $\rho \in (X \times \Delta)^*$, such that $T_o(Q_0,\rho)!$, then since $\rho \in P_{X\Delta_o}(\mathcal{L}_{io}(G))$, there should exists  $w \in \mathcal{L}(G)$, and $\alpha \in P(O(w))$ such that $P_X(\rho)=w$,  $P_{\Delta_o}(\rho)=\alpha$, and $\Tilde{Q}(w,\alpha) \neq \emptyset$. In addition,
following Definition  \ref{def:RCSO_observer}, $\Tilde{Q}(w,\alpha)$ and   $T_o(Q_0,\rho)$ provides the same estimated states, meaning, for any $q \in \Tilde{Q}(w,\alpha)$, we have $\hat{q}=T_o(Q_0,\rho)$ with $q \in \hat{q}$. Therefore, if $G$ is RCS-opaque, then $\Tilde{Q}(w,\alpha) \not \subseteq Q_s$ which implies $\hat{q} \not  \subseteq Q_s$. 

Sufficiency: here we show if for all $\hat{q} \in Q_o$, we have $ \hat{q} \not \subseteq Q_s $ then $G$ should be RCS-opaque. We prove this part by contradiction. Let's assume $G$ is not RCS-opaque that implies there should exists a  $w \in \mathcal{L}(G)$ such that $\Tilde{Q}(w,\alpha) \subseteq Q_s$ for some $\alpha \in P(O(w))$. Therefore, similar to the necessary part, we know $\Tilde{Q}(w,\alpha)$ and $T_o(Q_0,\rho)$ with $P_X(\rho)=w$ and $P_{\Delta_o}(\rho)=\alpha$, provide the same estimated states. This implies,  we have the observer state $\hat{q}=T_o(Q_0,\rho)$ that $\hat{q} \subseteq Q_s$ which contradicts the first assumption.
\end{proof}

The following example illustrates the observer construction described above.

\begin{example}
Consider the open DES $G$ in Figure \ref{fig:example_NFT} with $Q_s=\{3\}$,  $\Delta_o=\{ \delta_1,\delta_2,a\}$, and  $\Delta_{uo}=\{ b\}$. The constructed observer for $G$ is shown in Figure \ref{fig:example_observer}. An edge label is in the form of $x,\delta$, where $x \in X$, and $\delta \in \Delta_{o,\epsilon}$. 
As it is shown in the Figure  \ref{fig:example_observer}, the secret state $\{3\}$ is reachable from the initial state in the constructed observer, indicating that  $G$ is not RCS-opaque.  \hfill $\Box$
\end{example}

\section{Conclusion} \label{sec:RCSO_con}
 In the conventional opacity formalism, the intruder is considered as a passive observer. In this paper, we studied opacity in the presence of an active intruder which beyond a passive observation,  is capable of manipulating the system behavior. In this setup, the active intruder can inject a certain input to the system and combine it with the observed system response to infer the secrets. We therefore introduced reactive opacity notions which characterize a property that regardless of how the intruder selects the input word, the system's secret property remains  indistinguishable from the non-secrets. We furthermore showed that all the proposed reactive opacity notions can be transformed into the RCSO.  Given a RCSO notion and a system modeled as NFT, we proposed an automata-based method to verify if the system respects RCSO requirements. 
 In the future works, we plan to study probabilistic reactive opacity for stochastic DESs.

\bibliographystyle{IEEEtran}        
\bibliography{RCSO.bib}

\end{document}